\documentclass[11pt,letterpaper]{article}

\usepackage{fullpage}
\usepackage{amsmath,amssymb,amsthm,dsfont}
\usepackage{tikz}
\usetikzlibrary{shapes.symbols,arrows}

\usetikzlibrary{shapes.symbols,arrows}
\usepackage{enumerate}

\usepackage{algorithm,algorithmic}

\newcommand{\In}{{\mathrm {In}}}

\newcommand{\IN}{\mathds{N}}

\newcommand{\dG}{\mathds{G}}

\renewcommand{\leq}{\leqslant}

\renewcommand{\geq}{\geqslant}

\newcommand{\Ker}{{\mathrm{Ker}}}
\newcommand{\Roo}{{\mathrm{Roots}}}
\newcommand{\CRoo}{{\mathrm{CRoots}}}

\renewcommand{\leq}{\leqslant}

\renewcommand{\geq}{\geqslant}
\newcommand{\Ina}{{\mathrm {In}}^{{\mathrm{a}}}}

\newcommand{\idG}{\overline{\mathds{G}}}

\newcommand{\eqdef}{\stackrel{\text{def}}{=}}
\newcommand{\p}{\delta}

\newcommand{\dGa}{\mathds{G}^{{\mathrm{a}}}}

\newcommand{\AGE}{\mathrm{AGE}}

\newcommand{\isc}{infinitely connected}

\newtheorem{thm}{Theorem}
\newtheorem{prop}[thm]{Proposition}
\newtheorem{lem}[thm]{Lemma}
\newtheorem{cor}[thm]{Corollary}
\newtheorem{defi}[thm]{Definition}

\newcommand{\ignore}[1]{}

\begin{document}

\title{MinMax Algorithms for Stabilizing Consensus}

 \author{%
 Bernadette Charron-Bost\textsuperscript{1} 
  \and Shlomo Moran\textsuperscript{2}}
 \date{\textsuperscript{1} CNRS, \'Ecole polytechnique, 91128 Palaiseau, France\\
 \textsuperscript{2} Department of Computer Science, Technion, Haifa, Israel 32000\\~\\
 \today
  }




\maketitle

\begin{abstract}
In the \emph{stabilizing consensus} problem, each agent  of a networked system  has an input value 
	and is repeatedly writing an output value; it is required that eventually all the output values  stabilize 
	to the same value which, moreover,  must be one of the input values.
We study this problem for a synchronous model with identical and anonymous agents 
	that are connected by a time-varying topology.
Our main result is a generic MinMax algorithm that solves the  stabilizing consensus problem in this model 
	when, in each sufficiently long but bounded period of time, there is an agent, called a root,  that can send 
	messages, possibly indirectly,  to all the agents.
Such topologies  are highly dynamic (in particular, roots may change arbitrarily over time)  
	and enforce no strong connectivity property (an agent may be never a root).
Our distributed MinMax algorithms require neither central control (e.g., synchronous starts)  
	nor any global information  (eg.,on the size of the network), 
	and are quite efficient  in terms of message size and storage requirements.

\end{abstract}

\section{ Introduction}\label{sec:intro}

There has been much recent interest in distributed control and coordination of networks consisting of 
	multiple mobile agents. 
This is motivated by the emergence of large scale networks with  no central control and time-varying topology. 
The  algorithms deployed in such networks ought to be completely distributed, using only local  information, 
	and robust against unexpected changes in topology, despite the lack of global coordination like synchronous starts.
	
A canonical problem in distributed control is the \emph{stabilizing consensus}  problem~\cite{AFJ06,DAMSS11,BCNPT16}:  
	each agent~$u$ starts with some initial value and repeatedly updates  an output variable~$y_u$ which eventually 
	stabilizes on the same input value.
The  stabilizing  consensus problem arises in a number of applications including eventual consistency in replicated databases
	(see eg.,~\cite{Vog09}), motion of autonomous agents \cite{VCBCS95}, 
	and blockchain agreement~\cite{Nak08,BMCNKF15}.
Similarly,  the stable computation of a predicate in the model of population protocols~\cite{ADFP04} may be seen as a variant
	of  stabilizing  consensus, in which the stabilized output value  is the truth value of  some  predicate of the multiset of
	initial values. 	

A stronger form of agreement is captured by the classical \emph{consensus} problem which differs from  stabilizing 
	consensus in the fact that all the output variables $y_u$ are write-once: when the agent~$u$  is aware 
	that agreement has been reached on some initial value~$\mu$,  it writes~$\mu$ in~$y_u$, in which case~$u$ 
	is said to \emph{decide on}~$\mu$.
Hence the discrepancy between  stabilizing  consensus and consensus typically lies in this additional requirement of 
	irrevocable decisions.	

On the other side, \emph{asymptotic consensus} is a classical weakening of  stabilizing  consensus: 
	in the case of initial values that are real numbers, agents are only 
	required to compute the same outcome \emph{asymptotically}.
In other words, the condition of  eventual stability on the variables~$y_u$ is replaced by the weaker one of  convergence.
Moreover the limit value is only required to be in the range of the initial values, 
	which prevents the applicability of asymptotic consensus to the class of problems where the limit value must be one of the
	initial values.


Although there is a plethora of papers on agreement problems in multi-agent systems, few are specifically 
	devoted to  stabilizing  consensus.
To the best of our knowledge, the problem has been first investigated by  Angluin, Fischer, and Jiang~\cite{AFJ06}. 
They studied solvabilty of stabilizing consensus in an asynchronous totally connected system where agents 
	have distinct  identifiers and  may experience various type of faults, focusing on Byzantine faults.
 This problem  has been studied later  in \cite{DAMSS11,BCNPT16} in the synchronous  \emph{gossip model}.
 These papers	propose  randomized stabilizing consensus algorithms with convergence times that are functions of the number 
	of possible input  values.  
 
The original consensus problem, with irrevocable decisions, has been the subject of much more study,  	
	specifically in the context of fault-tolerance and a fixed topology.
There is also a large body of previous work on consensus in dynamic networks. 
In the latter  works, agents are  supposed to start synchronously, to share global informations 
	on the network, and to have distinct identifiers~\cite{CG13}.
Moreover,	 topology changes are dramatically restricted~\cite{BRS12}, or  communication 
	graphs are supposed to be permanently bidirectional and connected~\cite{KMO11}. 
	
The asymptotic consensus problem has been also extensively studied 
	as it arises in a large variety of applications in automatic control or for the  modeling of natural phenomenas~\cite{VCBCS95}.
Averaging algorithms, in which   every agent repeatedly takes a weighted average of its own value and values 
	received from its neighbors, are the natural and widely studied algorithms for this problem.
One central result by Cao, Morse, and Anderson~\cite{CMA08a} is that every \emph{safe} averaging algorithm	
	-- that is, an averaging algorithm where positive weights are uniformly bounded away from zero --
	solves this problem with a continually rooted, time-varying topology, even if the set of roots and links 
	change arbitrarily. 

\noindent {\bf Contribution.}
The primary goal of this paper is the design of stabilizing  consensus algorithms for synchronous, fault free 
	networks of identical and anonymous agents connected by a time-varying topology
	without any guarantee of strong connectivity. 
It should be noted that while stabilizing  consensus is trivially solved by a gossip algorithm when the  time-varying 
	topology is eventually strongly connected, in the sense that for every pair of agents $u$ and $v$
	there always exists a  time consistent path from $u$ to $v$, there is no obvious solution
	in the case some nodes cannot receive information from part of the network.
In the absence of such a connectivity property,  synchronous starts cannot be simulated~\cite{CBM18}, 
	and hence tolerating asynchronous starts makes the problem even more challenging.
	
We start by introducing the notion of  \emph{kernel}  that models the set of \emph{root agents} able at any time 
	to send messages, possibly indirectly, to all other agents.
If  this can be achieved in at most~$T$ communication steps, then the topology is said to be  
	\emph{rooted with delay}~$T$.
 A time-varying topology with a non-empty kernel is  thus rooted with finite but  {\it a priori} unbounded  delays.
We first prove that  stabilizing  consensus is not solvable in the case of an empty kernel.
Then we show that in the case of a time-varying topology that is rooted with  bounded delay, 
	the stabilizing  consensus problem is solvable, even if the bound is unknown. 
	
For that, we introduce the  \emph{MinMax} update rules for the output variables $y_u$, 
 	and  then provide a distributed  implementation of these update rules that is efficient,  both
	in terms of message size and storage requirements.
The resulting distributed algorithms, called \emph{MinMax} algorithms, require no leader, no agent identifiers, 	
	and assume no  global knowledge of the network structure or size.
Moreover, they tolerate that agents join the system asynchronously.
We define the subclass of {\em safe} MinMax algorithms, and show that any such algorithm achieves 
	stabilizing  consensus if the topology is rooted with bounded delay.
As a corollary, we get that stabilizing consensus is solvable in any asynchronous and completely connected network
	 and a minority of faulty agents that crash or commit send omissions.
Finally,   we show that using  safe MinMax algorithms,  stabilizing  consensus is not solvable under the sole assumption
	of a non-empty kernel, i.e., the topology is rooted with finite but unbounded delays.

Another contribution of this work is the introduction of new notions that
	capture global properties of  {dynamic graphs}, like the {\em kernel},
	  the {\em integral}, the {\em limit  superior} of a dynamic graph, which
	  we believe to be useful for investigating other distributed problems
	in networked systems with time-varying topologies.
	
\section{ Preliminaries}\label{sec:model}
 
\subsection{The computational model}

We consider a networked system with a {\em fixed} set $V$ of $n$ agents.
Our algorithms assume anonymous networks in which agents have no identifiers
	and do not know the network size  $n$. 

We assume a round-based computational model in the spirit of the Heard-Of model~\cite{CBS09}. 
Point-to-point communications are organized into \emph{synchronized rounds}: each agent can send messages 
	to all agents and can receive messages sent  by some of the agents.
Rounds are communication closed in the sense that no agent receives messages in round~$t$ that are sent 
	in a round different from~$t$.
The collection of \emph{possible} communications (which agents can communicate to which agents) at each round $t$
	is modelled by a directed graph (digraph) with one node per agent.
The digraph at round~$t$ is  denoted $\dG(t)=(V,E_t)$, and is called the \emph{communication graph et round}~$t$. 
When dealing with just graph notions, we will use the term node rather than the one of  agent for an element of~$V$.
We assume a self-loop at each node in all these digraphs  since every agent can communicate with 
	itself instantaneously.	
The sequence of digraphs~$\dG=\left (\dG(t) \right )_{t\in\IN}$ is called a {\em dynamic graph}~\cite{CFQS11:TVG}. 
A \emph{network model} is any non-empty set of dynamic graphs.

In every {\em run} of an algorithm, each agent~$u$ is initially {\em passive}: 
	it  neither sends nor receives messages, and do not change its state. 
Then  it either becomes {\em active}  at the beginning of some round~$s_u\geq 1$, 
	or remains passive forever -- in which case we let $s_u = \infty$.
A run is {\em active} if all agents are eventually active. 	
			
At the beginning of its starting round~$s_u$,  the agent~$u$ sets up its local variables 
	and starts executing its program.
In round~$t \geq s_u$, $u$ sends messages to all agents, receives messages from all its incoming 
	neighbors in  the digraph $\dG(t) $ that are active, and finally goes to its next state applying
	a deterministic transition rule.
Then the agent~$u$ proceeds to round $t+1$. 
The number of the current round is  not assumed to be provided to the agents. 

The value of a local variable $x_u$ of $u$ at the \emph{end} of round~$t \geq s_u$  is denoted by $x_u(t)$. 
By convention, the value of $x_u(t)$ for $t<s_u$ 
	is defined as  the initial value of $x_u$. 
	
Since each agent is deterministic, a run  is entirely determined 
	by the initial state of the network,  the dynamic graph $\dG$,
	and  the collection  of the starting rounds.
For each run,    $\dGa(t) = (V, E^{{\mathrm{a}}}_t) $ denotes the digraph where $E^{{\mathrm{a}}}_t\subseteq E_t$
	is the set of edges  that are either self-loops\footnote{%
	To allow for simple notation, there are self-loops at all nodes of $\dGa(t)$, including those corresponding to 
	the passive agents at round~$t$.}
	or connecting  two  agents  that are active in round $t$.
The sets of $u$'s incoming neighbors  (in-neighbors) in the digraphs $\dG(t)$ and $\dGa(t)$
		are denoted	by $\In_u(t)$ and $\Ina_u(t)$, respectively.

\subsection{ Limits and integrals of dynamic graphs}

Let us first  recall that the {\em product} of two digraphs $G_1=(V,E_1)$ and $G_2=(V,E_2)$, denoted $G_1 \circ G_2 $, 
	is the digraph  with the set of nodes $V$ and with an edge $(u,v)$ if there exists~$w\in V$ such that $(u,w)\in E_1$ 
	and $ (w,v) \in E_2$.
For any dynamic graph $\dG$ and any integers $t' > t \geq 1$, we let $ \dG(t:t') = \dG(t) \circ  \dots \circ \dG(t')$.
By convention, $\dG(t:t)=\dG(t) $, and  when $0\leq t'<t$, $ \dG(t:t') $ is the digraph with only a self-loop at each node.

Given any dynamic graph $\dG$ and any scheduling of starts,  the sets of $u$'s incoming neighbors 
	(or in-neighbors for short) in $\dG(t:t')$ and in $\dGa(t:t')$ are denoted by $\In_u(t  :t' )$ and $\Ina_u(t  :t' )$, respectively,
	and simply by $\In_u(t)$ and $\Ina_u (t)$ when $t' = t$.
Because of the self-loops,  all these sets  contain the node~$u$.
If  $t'<t$, then $\In_u(t  :t' ) = \Ina_u(t  :t' ) = \{ u \}$. 

If  $t  \leq t'$, then a $v {\sim } u$ \emph{path in the interval}  $[t,t' ]$ is any finite sequence 
	$ w_0 = v,w_1,\dots ,w_m = u$ with  $m = t' - t+1$ and 
	$(w_k,w_{k+1})$ is an edge of $\dG(t + k)$ for each $k = 0,\dots , m - 1$. 
Hence there exists a $v {\sim } u$ path in the interval $ [t ,t'] $ if and only if $(v,u)$ is an edge of $\dG(t:t')$, 
	or equivalently $v \in  \In_u(t  : t' )$.
	
By extension over the infinite interval $[t, \infty)$, we define the digraphs
	$$ \dG( t :\infty )=  \left ( V,     \cup_{t'\geq t} E(\dG (t:t')) \right), ~~
		 \dGa( t :\infty )=  \left ( V,     \cup_{t'\geq t} E(\dGa (t:t')) \right), $$
	and denote by $\In_u(t:\infty) $ and  $ \Ina_u(t:\infty) $ the sets of $u$'s in-neighbors in these two 
	digraphs, i.e., 
	$$\In_u(t:\infty)=\cup_{t'\geq t}\In_u(t:t'), ~~\Ina_u(t:\infty)=\cup_{t'\geq t}\Ina_u(t:t')    .$$
The dynamic graph $\idG$, defined by $\idG(t) = \dG  ( t :\infty )  $, is called the \emph{integral} of $\dG$.

The \emph{limit superior} of $\dG$,  denoted by $\dG(\infty)$, is defined as the digraph 
	$  \dG(\infty)=(V,E_\infty) $,
	where $E_\infty$ is the set of edges that appear in an infinite number of digraphs $\dG(t)$, namely,
	$$ E_\infty = \{ (u,v) \in V\times V : \forall t, \exists t'\geq t, \ (u,v) \in E(\dG(t')) \}  .$$
In particular, the digraph $\idG(\infty)$  is the limit superior of $\idG$.

\begin{prop}\label{prop:idGstab}
If $\dG$ is a dynamic graph with a permanent self-loop at each node, then
	 $\idG$ eventually stabilizes to $\idG(\infty)$, i.e., there is a positive integer~$s$ such that
	$$ \forall t \geq s, \  \ \idG(t) = 	\idG(\infty)    . $$
\end{prop}

\begin{proof}
Because of the self-loops, every edge of $\idG(t+1)$ is an edge of $\idG(t)$. 
Hence the dynamic graph $\idG$ eventually stabilizes to some digraph $\idG(s)$, i.e., there is a positive integer~$s$ such that
	$$ \forall t \geq s, \ \idG(t) = 	\idG(s)    . $$
Hence all edges in $\idG(s)$ are edges of~$\idG(\infty) $.

Conversely, by definition of the limit superior, any edge of~$\idG(\infty) $ appears in some digraph $\idG(t)$ with 
	$t \geq  s$, and since $\idG(t)=\idG(s)$,  is also an edge of $\idG(s)$.
\end{proof}

Let us recall that a digraph $G$ is transitively closed if any edge of $G \circ G$ is an edge of $G$. 
In the case $G$ has a self-loop at each node, this is equivalent to $G \circ G= G$. 
The transitive closure of $G$, denoted by $  {G}^{ +} $, is the minimal transitively closed digraph
	that contains all edges of $G$.
	
\begin{thm}\label{thm:idg,dg}
If $\dG$ is a dynamic graph with permanently a self-loop at each node, then  
	$\idG(\infty) $  is the transitive closure of $\dG(\infty) $, namely,
	$$ \idG(\infty) =   [\dG(\infty) ]^{ +}    . $$
\end{thm}

\begin{proof}
First we prove that  $\idG(\infty) $ is transitively closed. 
Since every edge of  $ \dG(\infty) $ is also an edge of  $\idG(\infty) $, this will show that 
	$ [\dG(\infty) ]^{+}    \subseteq  \idG(\infty) $.
For that, let $s$ be the index from which $\idG$ stabilizes (cf. Proposition~\ref{prop:idGstab}), 
	and let $(u, v)$ and $(v, w)$ be two edges of  $\idG(\infty) $. 
Since $ \idG(\infty) = \idG(s)$, there exists an index $t \geq  s $ such that $ (u, v) $ is an edge in $\dG(s : t)$. 
Since $ \idG(\infty)  =  \idG(t+1) $, there exists an index $ t'  >  t$ such that $ (v, w) $ is an edge in $\dG(t+1 : t')$. 
It follows that $(u, w) $ is an edge in $\dG(s : t' )$, and hence $(u, w) $ is an edge in $ \idG(s) = \idG(\infty) $.

We now prove the reverse inclusion; let $(u, v) $ be an edge of $ \idG(\infty) $. 
Since there are finitely many edges that appear finitely many times in $\dG$, 
	there is an index $ r $ such that for all $ t  \geq  r$, any edge in $\dG(t) $ is an edge in $ \dG(\infty) $. 
Let $ t = \max (s, r )$. 
By Proposition~\ref{prop:idGstab}, $(u, v)$  is an edge of $ \idG(t)  = \idG(\infty)$, i.e., there exists an index 
	$ t'  \geq  t $ such that $ (u,v) $ is an edge of $ \dG(t : t'  )$. 
In other words, there is a $u \sim v $ path in the interval $ [t,t' ] $; 
	let  $ w_0 = u , w_1, w_2, \dots , w_{t' - t +1} = v  $ be such a path. 
Since $ t \geq   r $, each edge in this  path is an edge in $ \dG(\infty) $, which shows that $ (u, v) $ is an edge 
	of the transitive closure of $ \dG(\infty)  $,  namely $  [\dG(\infty) ]^{+}  $. 
 \end{proof}

\subsection{Roots, central roots, and kernels}

A  node~$u$ is  {\em  a root of } the digraph $G =(V,E) $  if for every node  $v \in V$, 
	there is  a path from $u$ to $v$ in $G$, and ~$G$ is said to be \emph{rooted} if it has at least one root.
Node~$u$ is  {\em  a central root of}~$G$ if for every node  $v \in V$,  $(u,v)$ is an edge of~$G$.
The set of~$G$'s roots and the set of~$G$'s central roots  are denoted by $\Roo(G)$ and $\CRoo(G)$, respectively.

The {\em  kernel} of a dynamic graph $\dG$, denoted by $\Ker(\dG)$, is  defined as
	$$ \Ker(\dG) =   \left \{ u \in V \mid \forall t\geq 1, \forall v \in V,  \exists t'\geq t :  (u,v) \in E(\dG(t:t')) \right\} $$
	 or equivalently,                     
	 $$ \Ker(\dG) =    \cap_{t\geq 1}  \CRoo(\idG(t)) . $$

\begin{prop}\label{pro:ker}
If $\dG$ is a dynamic graph with permanently a self-loop at each node, then 
	$$  \Ker(\dG) =  \CRoo \left( \idG(\infty) \right) = \Roo\left(\dG(\infty)\right)  .$$
\end{prop}

\begin{proof}
Because of the self-loops, $\CRoo\left ( \idG(t+1) \right) \subseteq \CRoo\left ( \idG(t) \right) $,
	which by Proposition~\ref{prop:idGstab} implies that  
	$$ \Ker(\dG) = \cap_{t\geq 1}  \CRoo(\idG(t)) = \CRoo \left( \idG(\infty) \right)   . $$
By Theorem~\ref{thm:idg,dg}, the digraph $\idG(\infty) $  is the transitive closure of $\dG(\infty) $, 
	and so 
	$$ \Roo(\dG(\infty)) = \CRoo(\idG(\infty))  $$
	which completes the proof.
\end{proof}

The dynamic graph  $\dG$ is said to be \emph{\isc} if  $\Ker(\dG) =V$, 
	or equivalently, by Proposition~\ref{pro:ker}, if  $\dG(\infty)$ is strongly connected.

\subsection{Bounded delay rootedness}

A dynamic graph~$\dG$ is said to be \emph{permanently rooted} if  all the digraphs $\dG(t)$ are rooted. 
This notion naturally extends as follows: 

\begin{defi}
A dynamic graph $\dG$ is \emph{rooted with delay} $T $ if 
	for every positive integer $t $, $\dG(t : t+ T - 1)$ is rooted. 
$\dG$ is {\em  rooted with a bounded delay} if it is 
	rooted with delay~$T$ for some fixed positive integer~$T$.
\end{defi}

Any dynamic graph that is rooted with a  bounded  delay 
	has a non-empty kernel, i.e., there are nodes 	that are central roots of all the digraphs $\idG(t) $.
Proposition~\ref{prop:boundpath} below shows that these nodes are actually central roots over bounded length intervals.
	
\begin{prop}\label{prop:boundpath}
If $\dG$ is  a dynamic graph  that is  rooted with delay $T$, then there exists a positive integer~$s$
	such that 
	$$\forall t\geq s,\forall u\in V,~\In_u \big( t : t + T(n- | \Ker( \dG) | ) \big) \cap \Ker( \dG) \neq \emptyset .$$ 
\end{prop}

\begin{proof}
For simplicity, we assume that $T= 1$; the general case can be  easily reduced to
	the case $T=1$ by considering the dynamic graph $\dG_T$ defined by $\dG_T (t) = \dG \big( (t-1)T + 1 : tT \big) $
	that is  rooted with delay one.

Let $s$ be a positive integer such that for all $t \geq s$, every edge of~$\dG(t)$  	is also an edge of the digraph~$ \dG( \infty) $,
	i.e.,  $E_t\subseteq E_\infty$.
 Then we have that
	\begin{equation}\label{eq:blush}
	 \forall t\geq s, \  \  \Roo(\dG( t ))  \subseteq \Ker(\dG )   .
	 \end{equation}
	
For any non-negative integer~$i$, let us now introduce the set $U_i$ of nodes that are outgoing neighbors 
	in the digraph~$\dG(t : t+i)$ of some of the nodes in~$\Ker (\dG)$ .	
Because of the self-loops, we have that $ \Ker(\dG) \subseteq U_{0}$ and
	$U_{i} \subseteq U_{i + 1}$.
We now show that either $ U_{i } = V $ or $U_{i} \subsetneq U_{i + 1}$.

For that, assume that $ U_{i } \neq V $.
Let $u \notin U_i$, and let $v$ be a root of the digraph $\dG( t+i +1)$;	
	hence there exists a path from~$v$ to $u$ in  $\dG( t+i +1)$.
From (\ref{eq:blush}) and the above inclusions, we derive that 
	$$ v \in \Roo(\dG( t + i + 1 ))  \subseteq \Ker(\dG ) \subseteq  U_0 \subseteq U_i  .$$
Thereby, there are two consecutive nodes~$x$ and~$y$ in the $v {\sim } u$ path such that
	$x \in U_i$ and $y \notin U_i$.
By construction,~$y$ is an outgoing neighbor of~$v \in \Ker (\dG)$ in the digraph~$\dG(t : t +i +1)$,
	and thus $y \in U_{i + 1}$.
In conclusion, $y \in U_{i+1} \setminus U_i$, which shows that $U_{i} \neq U_{i + 1}$.

It follows that $ | U_i | \geq \min ( n ,  k + i) $ where $k = | \Ker( \dG) |$, and hence
	$ U_{n-k} = V  $.
Thus for every node $u \in V$, it holds that $\In_u \big( t : t + T(n-k) \big) \cap \Ker( \dG) \neq \emptyset $,
		as required.
\end{proof}

\section{The  Stabilizing Consensus Problem}

Let ${\cal V}$ be a totally ordered set and let $A$ be an algorithm in which each agent $u$ has an input value 
	$ \mu_u \in {\cal V}$ and 
	an output variable $y_u $ initialized to $\mu_u$. 
The algorithm $A$   \emph{achieves stabilizing consensus} in an
	active run  with the initial values $(\mu_u)_{u\in V}$ if the following properties hold:
 
\begin{description}
	\item[Validity.] At every round $t$ and  for each agent~$u$, there exists some agent~$v$ such that $y_u(t)=\mu_v$.
	
	\item[Eventual agreement.] There exists some round~$ s $ such that
	$$ \forall t\geq s, \  \forall u,v \in V, \  y_v(t) = y_u(s)   .$$
	
\end{description}
The common limit value of the variables~$y_u$ is called the \emph{consensus value}.
The algorithm $A$ is said  \emph{to solve the  stabilizing consensus problem in a network model}  ${\cal G}$ 
	if it achieves  stabilizing consensus in each of its active runs with a dynamic graph in ${\cal G}$.

\section{A Necessary Condition for Stabilizing Consensus}

Next we show a necessary condition 
for the stabilizing consensus problem to be solvable.

\begin{thm}\label{thm:impossibility1}
There is no algorithm that solves  stabilizing consensus in a network model containing a 
	dynamic graph with an empty kernel.
\end{thm}
\begin{proof}
Let $\dG$ be any dynamic graph with an empty kernel, and let $s$ be an index such that 
	for every $t \geq s $, we have  $E_t \subseteq E_\infty$.
Let us consider the acyclic digraph formed with the strongly connected components of $\dG(\infty)$, 
	called the condensation graph of $\dG(\infty)$,  and let us recall that the condensation graph of a 
	non-rooted digraph contains at least two source nodes, i.e.,  two nodes with no incoming edges  
	(see e.g. \cite{EvenBook}). 
From Propositions~\ref{pro:ker}, we derive that the condensation graph 
	of $\dG(\infty)$ has at least two source nodes $U_0$ and $U_1$.
Hence from round~$s$, none of the agents corresponding to the nodes in~$U_0$ (resp. $U_1$) are reachable 
	from the agents corresponding to the nodes in~$U_1$ (resp. $U_0$).

For the sake of contradiction, assume that there exists an algorithm $A$ that achieves   stabilizing consensus 
	in all the runs with the dynamic graph~$\dG$.
Consider now any run of the algorithm~$A$ in which all agents  start  at round $s$, 
	and all agents of the strongly connected components $U_0$ and $U_1$ have the input values 0
	and 1, respectively.
Because of the validity property, all the agents in $U_0$ must set their output values permanently to 0. 
Similarly, all the agents in $U_1$ must set their outputs permanently to 1,
	which shows that the eventual agreement property is violated in this run.
\end{proof}

\section{MinMax Algorithms}
 
In this section, we define the class of \emph{MinMax} algorithms by the type of update rules for
	the variables $y_u$.
The way MinMax algorithms can be implemented in our computing model will be addressed in Section~\ref{sec:implementation}.

We  start with an informal description of these algorithms.
As a first step, consider the \emph{Min algorithm}, in which
 	 each agent $u$ has an output variable $x_u$  which is repeatedly set to the minimum input value $u$ has heard of.  
It is easy to see that on  dynamic graphs  that  are \isc,  the values of all  $x_u$ variables eventually  stabilize on the minimum input value.

	
When the Min algorithm is applied on an arbitrary dynamic graph,  $x_u$ eventually stabilizes on the minimum input value received by $u$, to be denoted by  $m^*_u$: 
	$$  m^*_u\eqdef\min_{v\in \Ina_u(1:\infty)}\big( \mu_v~\big)  . $$ 
Hence there is an integer $\theta$    such that for  every round~$ t \geq \theta $ and every agent~$u$, 
 	it holds that $  x_u(t) = m^*_u $.
As  shown below in  Lemma~\ref{lem:kermin}, if $u$ is in $\Ker(\dG)$, then  $m^*_u=m^*$, where 
	$$m^*\eqdef\max_{v\in V}\big(m^*_v\big). $$

If the integer  $ \theta $ is given, then the following simple two phase scheme can solve stabilizing 
	consensus on a dynamic graph with a non-empty kernel: 
The first  phase consists of the first $\theta$ rounds in which the Min algorithm is applied. 
This allows for each agent~$u$ to compute the value $m^*_u$ in the variable $x_u$. 
In the second  phase  starting  at round $\theta+1$, for each agent $u$  the variable $y_u$   
	is repeatedly set  to the maximal $m^*_v$ value $u$ has heard~of.
	
Since $\theta$ is not given, we implement the above scheme by assigning to each agent~$u$   at each round $t$  
	an integer $\theta_u(t)\leq t$,  and  by computing the value $y_u(t)$, 
	with the first phase consisting of the interval $[1,\theta_u(t)]$.
	
For this procedure   to be correct, we need that eventually  $\theta_u(t)\geq \theta$ for each agent $u\in V$. 
This is the case if $\lim_{t\to\infty}\theta_u(t)=\infty$.
Assuming that each agent $v$ has computed the value of $m^*_v$ by round $\theta_u(t)$, we also need that 
	each agent~$u$   hears of some agent in the kernel during the round interval $[\theta_u(t) +1, t]$.
In conclusion, $\theta_u(t)$ must be chosen (1) large enough to ensure that each agent~$v$ has computed 
	$m_v^*$ by round~$\theta_u(t)$ and (2) small enough to guarantee that $u$  hears of some agent in the 
	kernel during the period $[\theta_u(t) +1, t]$.


A {\em MinMax rule} for  the variable~$y_u$ is an update rule of the form 
	\begin{equation}\label{eq:minmax}
	  y_u(t) =   \max_{v\in \Ina_u(\theta_u (t)+1 : t) } \, \left ( \min_{w \in \Ina_v(1: \theta_u  (t) ) } \left (y_w(0) \right) \right  )
 	  \end{equation}
	where    $\theta_u (t)$ is any integer in the interval~$ [0,   t ]$.
A \emph{MinMax algorithm} is an algorithm in which for each agent $u$ and each round $t$, the value  of   $y_u$ is updated   
	by a MinMax rule.  
It is determined  by  the way the integer-valued  functions $\theta_u $, 
	called \emph{cut-off functions}, are chosen.  

 %
%
%

We now prove the basic property on which our strategy  relies.

\begin{lem}\label{lem:kermin}
In any active run, if $u$ is in $\Ker(\dG)$, then for every agent $v$   it holds that 
	$\Ina_u(1:\infty)\subseteq\Ina_v(1: \infty)$, and $ m^*_u = m^*$.
\end{lem}

\begin{proof}
Let $w$ be an arbitrary agent in $\Ina_u(1:\infty)$, and let $t_0\in\IN$ be such that  $w\in\Ina_u(1:t_0)$ and  
  all agents are active at round $t_0$. 
Since $u$  is in $\Ker(\dG)$, there is some round $t_1>t_0$ such that  $u\in\In_v(t_0+1:t_1)=\Ina_v(t_0+1:t_1)$. 
This implies that  $(w,v)$ is an edge of  $\dGa(1:t_1)$, and hence   $w\in\Ina_v(1: \infty)$.

From the definition of $m^*_u$, it follows that  $m^*_u \geq m^*_v $ for every agent $v $,
	and hence $m^*_u \geq m^* $.
By definition of~$m^*$, it holds that $m^*_u  \leq m^*$.
Therefore we have that $m^*_u  = m^*$ as required.
\end{proof}

\section{Safe MinMax Algorithms for Stabilizing Consensus}
 
We now define the subclass of \emph{safe MinMax algorithms}, and  
	present properties of dynamic graphs  guaranteeing  that safe MinMax algorithms 
	always stabilize on the value~$m^*$.

 \subsection{Definition of safe  MinMax algorithms }
 
 Let $m_u(t)$ be the minimal input value that $u$ has heard of by round $t$, i.e.,
	$$m_u(t) \eqdef  \min_{v\in \Ina_u(1:t)}    (\mu_v)     .$$ 
Using this notation,   the update rule~(\ref{eq:minmax}) can then be rewritten into
 \begin{equation}\label{eq:minmax2}
 y_u(t) =   \max_{v\in \Ina_u \left(\theta_u (t)+1 : t \right) } \,  \left ( m_v (\theta_u (t) )  \right)   .
 \end{equation}
 
By definition, the sequence $\big ( m_u(t) \big)_{t\geq 1}$ is non-increasing and  lower-bounded by~$m_u^*$. 
Thus it stabilizes to some limit value at some round denoted~$t_u$.
We let $t^* = \max\{t_v:v\in V\}$. 

\begin{lem}\label{lem:limit}
For each agent~$u$, $\lim_{t \rightarrow \infty} m_u(t) = m_u^*$.
\end{lem}

\begin{proof}
By definition of $m_u^*$, there exist some agent $v$ and some round $t$ such that
	$ v \in \Ina_u(1:t)$ and $ m_u^* = \mu_v $.
Hence, we get $m_u(t) \leq  \mu_v$, and so $m_u(t) \leq  m_u^*$.
Since $\Ina_u(1:t) \subseteq \Ina_u(1:\infty)$, we have $m_u(t) \geq  m_u^*$, and the lemma follows.
\end{proof}

Now consider an arbitrary agent $u$. 
Our goal is to set restrictions on the cut-off function~$\theta_u $ enforcing that eventually $y_u(t)= m^*$.
The first restriction   is that for all large enough $t$, 
  \begin{equation}\label{eq:first}
  \forall v\in V, \ ~	m_v(\theta_u(t))=m^*_v .
  \end{equation} 
Because the sequence $\big ( m_u(t) \big)_{t\geq 1}$ is stationary and by Lemma~\ref{lem:limit},  
	the condition~(\ref{eq:first}) is satisfied for all large enough $t$ if  $\lim_{t\to\infty} \theta_u(t)=\infty$. 
  	 
Assuming that (\ref{eq:first}) holds for some $t\in\IN$, we use Lemma \ref{lem:kermin} to
  	  show that if $\Ina_u  (\theta_u (t)+1 : t  ) $ contains an agent from $\Ker(\dG)$ then $y_u(t)=m^*$ as needed.   	 
In a large class of dynamic graphs with non-empty kernels, the latter condition is  satisfied whenever 
	$t-\theta(t)$ is larger than some constant (which may depend on the given dynamic graph). 
So our second restriction is that $\lim_{t\to\infty}   t - \theta_u(t) = \infty $. 

The above discussion leads to the following definition: A MinMax algorithm is {\em safe} if in each of its active runs,
	it holds that 
	\begin{equation}\label{eq:safe}
	\forall u \in V,~	\lim_{t\to\infty}\theta_u(t)=  \lim_{t\to\infty}  t - \theta_u(t) = \infty .
	\end{equation}

\begin{thm}\label{thm:minmax}
Any active run of a safe MinMax algorithm on a dynamic graph $\dG$  achieves stabilizing consensus 	
	if  there is a positive integer~$s$ such that 
	\begin{equation}\label{eq:ker}
	\forall t\geq s,\forall u\in V,~	In_u \big(\theta_u(t) + 1 : t \big)\cap \Ker (\dG)\neq\emptyset.
	\end{equation}
\end{thm} 

\begin{proof}
Without loss of generality, assume that $s\geq t^*$  and $s  \geq \max_{v\in V} (s_v)$. 
Let us consider an arbitrary agent~$u$.
Since the algorithm is safe, there is a positive integer $t_0$ such that  
	$\theta_u(t) \geq s$ for all $t\geq t_0$. 
Then for $t\geq t_0$, Equation~(\ref{eq:minmax2}) can be rewritten into
	\begin{equation}\label{eq:y=max}
	 y_u(t) =   \max_{v\in \In_u \left(\theta_u (t)+1 : t \right) } \,    \left ( m^*_v  \right )    .
	 \end{equation} 
 This immediately implies that $ y_u(t)  \leq m^* $.

We also obtain from 
(\ref{eq:y=max}) that for every agent~$v \in \In_u(\theta_u(t) + 1 : t)$,
	it holds that $ y_u(t) \geq m^*_v$.
By  (\ref{eq:ker}), the set $\In_u(\theta_u(t) + 1 : t)$ contains at least one agent $v$ in $\Ker(\dG)$.
Lemma~\ref{lem:kermin} implies that  $ m^*_v  = m^*$,
	and so $ y_u(t)  \geq m^* $. 
	
We conclude that for all $t>t_0$,  it holds that $y_u(t)=m^*$, i.e., the run achieves stabilizing consensus on $m^*$.
\end{proof}

As observed, the Min (or Max) algorithm solves stabilizing consensus 
	in  any \isc~dynamic graph. 
Simple examples show that this is not the case for some MinMax algorithms.
However, as a direct consequence of Theorem~\ref{thm:minmax}, we obtain the following result.

\begin{cor}\label{cor:stronglyconnected}
Every safe MinMax algorithm solves the stabilizing consensus problem
	in the network model of 
	infinitely connected dynamic graphs.
\end{cor}

As for dynamic graphs that are rooted with a bounded delay, the combination of 
	Proposition~\ref{prop:boundpath} and Theorem~\ref{thm:minmax} yields the following corollary.
	
\begin{cor}\label{cor:bdrooted}
Every safe MinMax algorithm solves the stabilizing consensus problem
	in the network model of dynamic graphs that are  rooted with a bounded delay.
\end{cor}

Interestingly, Corollaries~\ref{cor:stronglyconnected} and~\ref{cor:bdrooted} are the analogs
	for stabilizing consensus and MinMax algorithms of the fundamental solvability results by Moreau~\cite{Mor05}
	and by Cao, Morse, and Anderson~\cite{CMA08a} for asymptotic consensus and averaging algorithms.
Observe, however, that Corollary~\ref{cor:stronglyconnected} holds for all infinitely connected dynamic graphs
	while the Moreau's theorem requires the communication graph to be bidirectional  at every round.
		
\subsection{A limitation of safe MinMax algorithms}

A natural question raised by Theorem \ref{thm:minmax} is whether safe MinMax algorithms solve 
	 stabilizing consensus for every dynamic graph with a non-empty kernel. 
We show that this is not the case, and first  establish the following property of  safe MinMax algorithms.

\begin{lem}\label{lem:m*=y*}
In any run of a safe MinMax algorithm  in which stabilizing consensus is achieved, 
	the stabilizing consensus value is equal to~$m^*$.
\end{lem}

\begin{proof}
Let $\tilde{y}$ be the stabilizing consensus value, and consider a agent $u$ such that $m^*_u =m^*$. 
Observe that for all rounds~$t$  it holds that $y_u(t)\geq m_u(t)$, and for all large enough $t$ 
	we have $y_u(t)=\tilde{y}$ and $m_u(t)=m^*$. 
Thus we get that $\tilde{y}=y_u(t)\geq m_u(t)=m^*$ for all large enough $t$, and hence $\tilde{y}\geq m^*$.   

Conversely, let us consider an arbitrary agent~$u$.
Since the algorithm is safe, $\theta_u(t)>t^*$ if~$t$ is  large enough.
For each such~$t$, by (\ref{eq:minmax2}), we have 
	$ y_u (t) =   \max_{v\in \Ina_u \left(\theta_u (t)+1 : t \right) } \,    \left (m^*_v  \right) $,
	which implies that $ y_u(t) \leq m^*$.
The value of~$y_u$ stabilizes to~$\tilde{y}$,
	which shows that $ \tilde{y}  \leq m^*  $.
	
It follows that $\tilde{y} =m^*$, as claimed.
\end{proof}

%

\begin{thm}\label{thm:limitations}
	There is no safe MinMax algorithm that solves stabilizing consensus in the network model
	of dynamic graphs  with non-empty kernels.
\end{thm}

\begin{proof}
The argument is by contradiction: suppose that there is a safe MinMax algorithm~$A$
	that solves stabilizing consensus in the network model ${\cal G}_{nek}$
	of dynamic graphs over a fixed set~$V$ of $n \geq 2$ nodes and with non-empty kernels.
Let us denote $V = \{ u, v_1, \dots, v_{n-1} \}$, and consider the  following two digraphs $G$ and $H$ 
	with the set of nodes $V$:
	\begin{enumerate}
		\item $G$ is the directed chain $u, v_1, \dots, v_{n-1}$;
		\item $H$ is the directed chain $v_1, \dots, v_{n-1}, u$. 
	\end{enumerate}
We consider the active runs of~$A$ in which all the nodes start at round 1 and where all the input values are equal to 0,
	except the input value of the node~$u$ that is equal to~1.	
	
First, we consider the dynamic graph $\dG_0$ in which $ \dG_0 (t) = G$ at all rounds~$t$.
Clearly, we have that $\Ker(\dG_0) = \{u\} $,
	and the corresponding maximal value    $m^*_0$ is $1$.
By Lemma~\ref{lem:m*=y*} the consensus value in this run is equal to~1, i.e.,
	there exists some positive integer~$t_0$ such that for each round~$t\geq t_0$ and each node~$w$, 
	it holds that $y_w(t) = 1$.
In particular, $ y_{v_1}(t_0) = 1$.
	
We now consider the dynamic graph $\dG_1$ such that  $ \dG_1 (t) = G$ for~$1 \leq t \leq t_0$
	and $ \dG_1 (t) = H $ for~$ t >  t_0 $.
Clearly $ \Ker(\dG_1) = \{v_1\} $, 
	and for the corresponding run of~$A$, we have $m^*_1= 0$.
By Lemma~\ref{lem:m*=y*}, the consensus value in this run is~0, i.e.,
	there exists some positive integer~$t_1$ such that for each round~$t \geq t_0 + t_1$ and each node~$w$, 
	it holds that $y_w(t) = 0$.
In particular,~$y_{v_1}( t_0 + t_1) = 0$.
	
	
By repeating the above construction,  
	we determine an infinite sequence of positive 
	integers~$ (t_k)_{k\in \IN}$ and the dynamic graph~$\dG_{\infty}$ defined by
	$$ \dG_{\infty}(t) = \left\{ 
	\begin{array}{ll}
	G & \mbox{ if } 1 \leq t \leq t_0 \mbox{ or }  \ t_0 + \dots +  t_{2k-1}  + 1 \leq t \leq t_0 + \dots +  t_{2k }  \\
	H & \mbox{ if } \  t_0 + \dots +  t_{2k } + 1  \leq t \leq t_0 + \dots +  t_{2k +1}   . 	
	\end{array} \right. $$
We easily check that 
	$ \Ker \big( \dG_{\infty} \big) = \{u, v_1\} $,
	and for the corresponding run of~$A$, it holds that	
	$$y_{v_1} (t) = \left\{ 
	\begin{array}{ll}
	1 & \mbox{ if }  t = t_0 + \dots +  t_{2k} \\
	0 & \mbox{ if }   t = t_0 + \dots +  t_{2k +1}  . 	
	\end{array} \right. $$
In this run with a non-empty kernel, the sequence $\left ( y_{v_1} (t) \right )$ is not convergent, 
	which violates the eventual agreement property.
\end{proof}

Extending the analogy above pointed out, we may observe that a similar impossibility result for averaging 
	algorithms and asymptotic consensus is proved
	in~\cite{BHOT05} with a different collection of three node dynamic graphs. 

\subsection{Convergence time of safe MinMax algorithms}\label{sec:converge}

Contrary to  safe averaging algorithms that converge in at most  an exponential (in the size~$n$ of the 
	network) number of rounds with dynamic graphs that are permanently rooted~\cite{CMA08a, CBFN15},  
	the convergence time of safe MinMax algorithms with such dynamic graphs may be arbitrarily large:
	For instance, inserting the complete digraph at one round $t > t^*$ may result in  changing the value~$m^*$, 
	and so may require the MinMax algorithm to stabilize again.

Thus MinMax algorithms are highly unstable with respect to  --~even sporadic~-- topology changes.
However, if we restrict our analysis to a dynamic graph~$\dG$ formed with a fixed rooted digraph,
	safe MinMax algorithms converge much faster than  safe averaging algorithms.
To see that, assume all nodes start at round one, and let $K= \Ker(\dG)$.
It is not hard to see that within less than $  |K|$ rounds, each  $m_v$ with $v\in K$ stabilizes  to~$m^*$. 
Proposition~\ref{prop:boundpath} states that for each node $u $ and each round $t$ it holds that 
	$$ K \cap\In_u(t:t+n-|K|) \neq\emptyset   . $$ 
It follows that if $| K |  \leq \theta_u(t) + 1 \leq t-( n-| K |)$, then $y_u (t) =m^*$.

Since $\lim_{t\to\infty}  \theta_u(t) = \lim_{t\to\infty}   t - \theta_u(t) = \infty $, there exists a positive integer
	$\alpha_u $ such that if $t > \alpha_u$, then it holds that
	$$ \theta_u(t) \geq | K | - 1 \mbox{ and } t - \theta_u(t) \geq n - | K |  + 1   .$$
We conclude that the algorithm stabilizes by round $\max_{u\in V} (\alpha_u)$ rounds.
In particular, setting $\theta_u(t)=\lfloor t/2\rfloor$ guarantees convergence within $2n$ rounds.

The latter result can be interestingly compared with the exponential lower bound proved in~\cite{OT11, CBFN15}:
	the convergence time of any safe averaging algorithm is exponential in~$n$ on the fixed rooted topology
	 of a \emph{Butterfly} digraph.

\section{Efficient Distributed Implementation of  MinMax Algorithms}\label{sec:implementation}

In this section,  we discuss distributed implementations of  MinMax algorithms in our  computing model.  
Figure~\ref{fig:implementation} presents  a general, efficient distributed scheme for this  implementation, 
	which is applicable whenever the {\em difference functions}  defined by
	$$\p_u =   t-\theta_u $$ are locally computable. 
The exact nature of the $\p_u$ functions is  left unspecified  (line~\ref{line:update p_u}).
	
Observe that the cut-off function $\theta_u$ satisfies the inequalities $0\leq \theta_u(t)\leq t$ 
	if and only if  $\p_u$ satisfies the same inequalities. 
The inequalities $ 0 \leq \p_u (t) \leq t$ can be easily enforced 
	by having the agent $u$ implement the simple round counter $C_u$ defined by  $C_u(t) = t- s_u$. 
Indeed, the difference function $\p_u(t)= f\big( C_u(t) \big) $ satisfies these two inequalities when  $f$ is 
	any integer-valued function such that $ 0 \leq f (t) \leq t$. 
Besides we can choose $f$ so that the difference function 
	$\p_u(t)= f\big( C_u(t) \big) $ provides a \emph{safe} MinMax algorithm:
	for instance,  we may set $f (k) = \lfloor k/2 \rfloor $ or $f( k  ) = \lfloor \log k \rfloor   $.
 
A possible, but quite inefficient way for implementing MinMax algorithms  consists in using a {\em  full information} 
	protocol, in which at each round $t$ each active agent sends its {\em local view} at round $t-1$ to all other agents;  
	the local view of $u$ at round $t$ 
	for $t\geq s_u-1$   is a rooted tree with labeled leaves, denoted $T_u(t)$,  defined inductively as follows: 
First, $T_u( s_u-1)$ is a single vertex labelled by $\mu_u$. 	
Assume now that at round $t$, the agent $u$ receives $k$ messages 
	with the trees~$T_1, \cdots, T_k$.
Then   $T_u(t)$ is the tree consisting of a root with $k$ children on which the  trees $T_1,\ldots T_k$ are hanged.
 Using  $T_u(t)$, the agent~$u$ can then easily compute $y_u(t)$ corresponding to the cut-off point
 	$\theta_u(t) = t-\p_u(t)$.

The point of our  implementation is precisely to avoid the construction of the trees~$T_u(t)$.
For that, each agent $u$ maintains, in addition to $y_u$ and $\p_u$,   a variable $x_u$ with values in ${\cal V}$.
At each round~$t$, the agent $u$ sets  $x_u$ to the minimal input value it has heard of,
	i.e., $x_u(t)=m_u(t)$. 
	
We say that an input value $\mu$ is {\em relevant for agent $u$ at round~$t$} if there is  an agent 
	$v\in\Ina_u(t-\delta_u(t)+1:t)$ 	such that  $x_v(t-\delta_u(t))=\mu$. 
Thus, the agent $u$ needs to set $y_u$ to its maximal relevant value at each round, which is done as 
	explained below.
	
 Just to simplify notation, we assume that the set  ${\cal V}$ of all the possible initial values  is finite and given.
To determine the set of its relevant input values, 
	the agent~$u$ maintains a vector of integers  $\AGE_u$ such that  for each $\mu\in{\cal V}$, 
	 $\AGE_u[\mu](t)$ is  the minimal number of rounds, by $u$'s local view at round $t$, that have passed 
	 since the last time some agent $v$ had set $x_v$ to $\mu$. 
Thus $\mu$ is relevant for $u$ at round~$t$ if and only if $\AGE_u[\mu](t)\leq\p_u(t)$.

\begin{figure} [h]
	\small
	\noindent\rule{17cm}{0.4pt}
	\begin{algorithmic}[1]
		\REQUIRE{}
		\STATE $x_u \in {\cal V} $, initially $\mu_u$\label{line:initxu} ;  $y_u \in {\cal V}  $, initially $\mu_u$ ;
		  $\p_u \in \IN $, initially 0 ;  $\AGE_u \in (\IN\cup\infty)^m$, initially $\infty^m$\label{line:init age vector}%
		\STATE $\AGE_u [x_u] \gets 0$\label{line:ageinit}
		\WHILE { TRUE}  
		\STATE send $  \AGE_u  $  to all agents
		\STATE receive $  \AGE_{v_1}     \cdots,  \AGE_{v_\ell}   $ 
		\STATE for all $\mu\in{\cal V}, \AGE_u[\mu]\gets 1+\min_{1\leq i\leq \ell}(\AGE_{v_i}[\mu])$\label{line:update age}
		\STATE $x_u\gets\min  \{\mu:\AGE_u[\mu]<\infty\} $\label{line:xu}
		\STATE $\AGE_u[x_u]\gets 0$\label{line:age0}
		\STATE   $\p_u\gets \cdots$  \label{line:update p_u}
		\STATE $y_u \gets \max  \{\mu :\AGE_u[\mu]\leq \p_u\}   $ \label{line:update y_u}
		\ENDWHILE
		\end{algorithmic}
	\caption{Distributed implementation of  a MinMax algorithm with the cut-off functions $  \theta_u = t- \p_u $.  }
	\label{fig:implementation}
	\noindent\rule{17cm}{0.4pt}
\end{figure}

Now we  show that the algorithm corresponding  to the difference functions $\p_u$ is a MinMax algorithm 
	with the cut-off functions $\theta_u=t-\p_u $.
We start by two preliminary  lemmas.



\begin{lem}\label{prop:X}
For  any agent $u$ and any round $t\geq 1$, 
	$ x_u (t) =  m_u(t)$.
\end{lem}

\begin{proof}
This is an immediate consequence of the initialization of the variable $x_u$ (line~\ref{line:initxu}), 
	 of its update rule (line~\ref{line:xu}), and the fact that if $t<s_u$ then $x_u(t)=x_u(s_u-1)$.
\end{proof}

\begin{lem}\label{prop:AGE}
If the agent $u$ is active at round $t$, then for each integer $k \in \{0, \dots, t\}$, 
	$$ \AGE_u[\mu](t) \leq k \Leftrightarrow 
	\exists v\in \Ina_u(t - k +1 : t ) , \ x_v   (t-k  ) = \mu  .$$
\end{lem}

\begin{proof}
First, assume that there is an  agent~$v \in  \Ina_u(t - k +1 : t )$ such that 
	$x_v  (  t-k  ) = \mu $, and let $t_0=\max\{t-k, s_v-1\}$. 
Since for $t<s_v$, we  have set $x_v(t)=x_v(s_v-1)$, it always holds that
	$x_v(t-k)=x_v(t_0)=\mu   $.
Moreover,  we easily check that 
	$v $ is in  $ \Ina_u(t_0+1 : t )  $.
Hence there exists a $v {\sim} u$ path in the interval $ [ t_0+1, t ] $
	that we denote by $ v_0=v, v_1, \dots , v_{t - t_0} = u $.
Because of the update rule of the vectors~$\AGE_w$, we deduce step by step that
	$$ \AGE_{v_1} [ \mu ] (t_0+1) \leq 1, \  \AGE_{v_2}[ \mu ] (t_0+2)  \leq 2, \ \dots \ , \  \AGE_{u} [ \mu ] ( t ) \leq t-t_0    .$$   
The claim follows by observing that $  t-t_0 \leq k $.

We now show by induction on $k$ the following implication:
	$$
	 \AGE_u[\mu](t) \leq k \Rightarrow 
	\exists v \in \Ina_u(t - k +1 : t ) , \ x_v (   t-k  ) = \mu  .
	$$

\begin{description}
\item[{\it Base case:}]
 $ \AGE_{u} [ \mu ] ( t ) = 0   $.
By lines~\ref{line:update age} and~\ref{line:age0}, we deduce that $x_u (t) = \mu$.
Then the agent~$v=u$  is in $\Ina_u (t+1:t)$ with $x_v (t) = \mu$, as required.

\item [{\it Inductive step:}]
Assume that the above implication holds for some non-negative integer~$k$,
	and let $  \AGE_u[\mu](t) \leq k  + 1$.
Then either (a) $  \AGE_u[\mu](t) \leq k  $   or (b) $  \AGE_u[\mu](t) =  k  + 1$.

\begin{description}
\item [\rm{(a)}] By inductive assumption, there is some agent~$w \in \Ina_u(t - k +1:t ) $
	such that $ \ x_w  (   t-k ) = \mu $.
	Then we consider the two following cases:
	\begin{enumerate}
	\item If $x_{w}(t-k-1)=x_{w}(t-k)$, then we let $v = w$. 
	\item Otherwise, $x_{w}(t-k-1)\neq x_{w}(t-k)$, which means that 
	in round $t-k$, $x_{w}$ was set  to $x_v(t-k)$ for some agent $v  $ in $\Ina_{w}(t-k)$ by executing line \ref{line:xu}.  
Thus we have that $ x_{w}(t-k)=x_v(t-k-1) = \mu$. 
	\end{enumerate}
In both cases, the proof of  the claim in case (a)  is completed by noting that 
	since $\dGa(t-k:t) = \dGa(t-k) \circ \dGa(t-k +1:t) $, we have that
	  $v\in\Ina_u(t-k:t)$.
\item [\rm{(b)}] By line~\ref{line:update age}, there is some agent $w$  in  $ \Ina_u(t) $ such that 
	$\AGE_w [\mu](t-1) = k  $.
	The inductive hypothesis implies that there exists some agent~$v $ in  $ \Ina_w( t -k : t -1 )$
	such that $$ x_v (  t - 1 - k  ) = \mu    .$$
	Since $\dGa(t-k : t ) = \dGa(t - k: t-1) \circ \dGa(t)$, it follows that $ v \in \Ina_u( t -k : t  )$
	as required.
\end{description}
\end{description}
\end{proof}
\begin{thm}\label{thm:distimp}
 Any instance of the scheme in Figure \ref{fig:implementation} is a MinMax algorithm,  with cut-off functions
	given at each round~$t$ by $ \theta_u (t)  = t - \p_u(t) $.
\end{thm}
\begin{proof}
If the agent $u$ is active at round $t$, then  we have 
	$ y_u(t) =  \max \, \{ \mu \in {\cal V} : \AGE_u[\mu]\leq \p_u(t) \}  $ (cf. line~\ref{line:update y_u}).
From Lemma~\ref{prop:AGE}, it follows  that 
	\begin{equation}\label{eq:x_u(t)}
	 y_u(t) = \max  \, \{ \mu \in {\cal V} : \exists v \in \Ina_u( \theta_u  +1 : t),  \ x_v (  \theta_u  )  = \mu \}
	  \end{equation}
	where $ \theta_u = t - \p_u(t) $.
By Lemma~\ref{prop:X},  it holds  that 
	\begin{equation}\label{eq:y_u(t)}
	x_v(\theta_u) =  \min_{w \in \Ina_v(1: \theta_u) } \big ( \mu_w   \big )  . 
	\end{equation}
By Equations (\ref{eq:x_u(t)}) and (\ref{eq:y_u(t)}), we get that, with  $ \theta_u = t - \p_u(t) $,
	$$ 	      y_u(t) =  \max_{v \in    \Ina_u( \theta_u  +1 : t) } \left ( \min_{w \in \Ina_v(1: \theta_u)} \mu_w  \right) . $$                 
 \end{proof}

The resulting MinMax algorithms share the same key features as averaging algorithms, namely they 
	assume no leader, do not use identifiers,  and tolerate asynchronous starts.
Unlike averaging algorithms, they are not memoryless, but  are  space efficient in the sense that except the 
	$AGE_u$ counters, which are bounded by $\log (t)$, all other local variables are of bounded size.
Actually, 	
	the unbounded  counters $\AGE_u[\mu]$ 
	-- which imply unbounded storage capacities and unbounded bandwidth --
	 are the discrete counterpart of the infinite precision required in averaging algorithms.

\subsubsection*{Stabilizing consensus with  failures}

In the light of Corollary~\ref{cor:bdrooted} and Theorem~\ref{thm:distimp}, we now revisit the problem of stabilizing consensus
	in the context of benign failures.
In particular, we consider completely connected systems and
	the failure model of crashes or the one of send omissions.
Basically, the resulting communication graphs are not strongly connected,
	 and thus naive approaches (e.g., the Min algorithm) do not work.
	 
To tackle this problem, we propose a strategy consisting first in emulating synchronized rounds and then in using a 
	safe MinMax algorithm on top of this emulation.
Indeed, as demonstrated in~\cite{CBS09}, synchronized rounds with a dynamic communication graph 
	can be easily emulated in any such distributed system, be it synchronous or asynchronous, when the 
	network size is given:~synchrony assumptions and failures are captured 
	as a whole just by the connectivity properties of the dynamic graph.
Typically,  synchronized rounds with a dynamic graph that 
	is \emph{non-split} \footnote{%
	A digraph is non-split if any two nodes have a common in-neighbor.}
	 at each round  can be  emulated if a minority of  agents may crash or  fail by send omissions.
Since a non-split digraph is rooted,  
	 any safe MinMax algorithm on top of this emulation solves stabilizing consensus despite asynchrony and agent failures.

\begin{cor}\label{cor:failures}
The stabilizing consensus problem is solvable in an  asynchronous system with a complete topology, reliable links,
	and a minority of agents that crash or commit send ommissions.
\end{cor}

\section{Concluding Remarks}
In this paper, we studied the stabilizing consensus problem for  dynamic networks  with very few restrictions
	on the computing model and the network.
In particular, we did not restrict link changes, except for retaining  a weak connectivity property, namely rootedness
	over sufficiently long periods of time, captured by the condition of a non-empty kernel.
First we showed that this property is necessary for solving stabilizing consensus, and then proved that it is  
	nearly a sufficient property, in the sense that every safe MinMax algorithm 
	solves stabilizing consensus if the dynamic graph is rooted with a  bounded delay.
Our solvability results for stabilizing consensus and MinMax algorithms are actually the analogs  of the ones 
	 for asymptotic consensus and averaging algorithms.		
	
Our work leaves open several questions.
First, it would be interesting to study whether the stabilizing consensus problem remains solvable when 
	the dynamic graph is rooted with finite but unbounded delays.
That may lead to the design of  algorithms, other than MinMax algorithms,
	that solve stabilizing consensus with no strong connectivity. 
Another related question concerns convergence time.
As demonstrated in Section \ref{sec:converge}, the convergence time of any safe  MinMax algorithm is unbounded 
	even for a dynamic graph that is permanently rooted, i.e., rooted with delay one.
This raises the following question: 
	does there exist another class of stabilizing consensus  algorithms that reach consensus in 
	bounded time -- which might depend on the network size -- for this specific model of dynamic graphs?

\bibliographystyle{plain}

\end{document}